\documentclass[11pt]{article}

\usepackage[margin=1in]{geometry}       
\usepackage{lmodern}                    
\usepackage[T1]{fontenc}
\usepackage{microtype}                  

\usepackage{authblk}                    
\usepackage{natbib}                     
\usepackage{amsmath}
\usepackage{amssymb}
\usepackage{amsthm}
\usepackage{booktabs}
\usepackage{algorithm}
\usepackage{algorithmic}

\usepackage{tikz}

\newtheorem{instance}{Instance}
\newtheorem{theorem}{Theorem}
\newtheorem{lemma}{Lemma}
\newtheorem{proposition}{Proposition}

\newtheorem{corollary}{Corollary}

\theoremstyle{definition}
\newtheorem{definition}{Definition}

\newif\ifanonymize
\anonymizefalse            

\title{On Pareto-Optimal and Fair Allocations with Personalized Bi-Valued Utilities\thanks{This paper has been accepted in WINE'25.}}

\ifanonymize
  \author{Anonymous Authors} 
\else
  \author[1]{Jiarong Jin}
  \author[2]{Biaoshuai Tao}
  \affil[1]{Shanghai Jiao Tong University, {\tt jinjiarong@sjtu.edu.cn}}
  \affil[2]{Shanghai Jiao Tong University, {\tt bstao@sjtu.edu.cn}}
\fi

\date{}
\begin{document}
\maketitle

\begin{abstract}
We study the fair division problem of allocating $m$ indivisible goods to $n$ agents with additive \emph{personalized bi-valued utilities}. 
    Specifically, each agent $i$ assigns one of two positive values $a_i > b_i > 0$ to each good, indicating that agent $i$'s valuation of any good is either $a_i$ or $b_i$.
    For convenience, we denote the value \emph{ratio} of agent $i$ as $r_i = a_i / b_i$.
    
    We give a characterization to all the Pareto-optimal allocations. Our characterization implies a polynomial-time algorithm to decide if a given allocation is Pareto-optimal in the case each $r_i$ is an integer. 
    For the general case (where $r_i$ may be fractional), we show that this decision problem is coNP-complete. 
    Our result complements the existing results: this decision problem is coNP-complete for tri-valued utilities (where each agent's value for each good belongs to $\{a,b,c\}$ for some prescribed $a>b>c\geq0$), and this decision problem belongs to P for bi-valued utilities (where $r_i$ in our model is the same for each agent).

    We further show that an EFX allocation always exists and can be computed in polynomial time under the personalized bi-valued utilities setting, which extends the previous result on bi-valued utilities. 
    We propose the open problem of whether an EFX and Pareto-optimal allocation always exists (and can be computed in polynomial time).
\end{abstract}

\section{Introduction}
\label{sec:intro}
Fair division problem studies how to allocate a set of resources to a set of agents \emph{fairly} and \emph{efficiently}.
Since the seminal works of~\citet{Steinhaus48, Steinhaus49}, this problem has received significant attention from researchers in mathematics, economics, and computer science~\citep{b3d0372254394dcea5ead76308a0ccf4,amanatidis2023fair, BramsTa96,procaccia2013cake,liu2023mixed,robertson1998cake}.
In this paper, we study the fair division problem with \emph{indivisible goods} where $m$ goods/items are allocated to $n$ agents with heterogeneous preferences.
We consider the efficiency criterion \emph{Pareto-optimality} and the fairness criterion \emph{envy-freeness}.

\paragraph{Pareto-optimality.}
An allocation is Pareto-optimal if no alternative allocation exists that improves at least one agent’s utility without decreasing any other agent’s utility. 
While identifying a Pareto-optimal allocation is straightforward (e.g., allocating all items to a single agent is Pareto-optimal when all valuations are positive), \emph{verifying} whether a given allocation is Pareto-optimal is computationally challenging. 
This verification problem is clearly in coNP, since if an allocation is not Pareto-optimal, another allocation that improves the situation for some agents without harming any others can be a certificate. 
For general additive valuations, it is shown by \citet{de2009complexity} that this decision problem is coNP-complete. 

Given the hardness of the general problem, it is natural to study restricted utility classes. A particularly well-studied case involves valuations with only a few distinct values. For \emph{bi-valued utilities}, where each agent's value for a good is either $a$ or $b$, \citet{aziz2019efficient} show that deciding Pareto-optimality can be done in polynomial time. In contrast, for \emph{tri-valued utilities}, where each value lies in $\{a, b, c\}$, the problem becomes coNP-complete~\citep{aziz2019efficient}.

\paragraph{Envy-freeness and its relaxations.}
Envy-freeness~\citep{Foley67,varian1973equity} is one of the most prominent fairness criteria considered in the past fair division literature. 
An allocation is envy-free if every agent weakly prefers her own allocated bundle over the bundle of anyone else.
In other words, each agent does not \emph{envy} any other agent.
However, an envy-free allocation may not always exist in the indivisible items setting (e.g., when one item is allocated between two agents).
Thus, relaxations of envy-freeness have been considered, and the two most notable ones are \emph{envy-freeness up to one item} (EF1)~\citep{budish2011combinatorial} and \emph{envy-freeness up to any item} (EFX)~\citep{caragiannis2019unreasonable}.
An allocation is EF1 if, for any pair of agents $i$ and $j$, \emph{there exists} an item in agent $j$'s allocated bundle whose removal will prevent agent $i$ from envying agent $j$.
An allocation is EFX if, for every pair of agents 
$i$ and $j$, the removal of \emph{any single} item from agent $j$'s bundle ensures that agent $i$ no longer envies agent $j$.
It is clear that EFX is stronger than EF1.
While it is widely known that an EF1 allocation always exists~\citep {budish2011combinatorial,lipton2004approximately}, it is an open problem if EFX allocations always exist in general.
We know that EFX allocations exist for two and three agents~\citep{plaut2020almost,chaudhury2020efx}, but the existential problem is already open for four agents~\citep{berger2022almost}.

EFX allocations have also been explored under restricted utility profiles where each agent assigns only a few distinct values to the items.
In particular, for tri-valued valuations, while the existence of EFX allocations remains open, finding even approximate solutions in this setting can be quite challenging. 
For bi-valued valuations, \citet{amanatidis2021maximum} shows that an EFX allocation always exists and can be computed in polynomial time.
Moreover, it is also shown in the same paper that an allocation maximizing the \emph{Nash welfare} (the product of all agents' utilities) can simultaneously achieve EFX and Pareto-optimality.
Unfortunately, this approach is unlikely to provide a polynomial-time algorithm to compute an EFX and Pareto-optimal allocation, as it is later known that computing the maximum Nash welfare allocation is APX-hard for bi-valued valuations~\citep{akrami2022maximizing,fitzsimmons2024hardness}.
Later, \cite{garg2023computing} provide a polynomial-time algorithm for this task by using a different approach.
The current state-of-the-art result for bi-valued valuations is provided by \cite{bu2024best}, who present a randomized polynomial-time algorithm achieving Pareto-optimality and strong fairness guarantees.

\paragraph{Personalized bi-valued utilities.}
In summary, checking whether an allocation is Pareto-optimal and finding an EFX allocation are both tractable for bi-valued utilities, but both problems become challenging for tri-valued utilities. 
In this paper, we consider a utility class---\emph{personalized bi-valued utilities}---that is more general than bi-valued utilities.
We say that a valuation profile is personalized bi-valued if for each agent $i$, there exist two values $a_i,b_i$ with $a_i>b_i>0$ such that agent $i$'s valuation on each item can be either $a_i$ or $b_i$. 
The class of bi-valued utilities then becomes a special case with $a_1=\cdots=a_n$ and $b_1=\cdots=b_n$.
Our work pushes the frontier of fair division by analyzing the tractability of Pareto-optimality verification and EFX allocation existence in this more general setting.

\subsection{Our Results}
\label{sec:ourresults}
Our first set of results concerns the computational complexity of deciding whether an allocation is Pareto-optimal.
We first provide a neat characterization for Pareto-optimal allocations in the case where each $r_i$ is an integer (Theorem~\ref{thm:main_po}).
Given this characterization, we show that if each $r_i$ is an integer, checking whether an allocation is Pareto-optimal is in P (Corollary~\ref{cor:main}).
We finally show that for general personalized bi-valued utilities, where $r_i$ may be fractional, the problem is coNP-complete (Theorem~\ref{thm:coNP-completeness}). 

To clearly describe our characterization, we first introduce two typical item-exchange scenarios that yield Pareto-improvements. 
These scenarios differ in the way agents exchange items that they value differently. 

\begin{itemize}
    \item[Type I:] ``small-large exchange'': Agent $i$ gives a good she values as \emph{small} to agent $j$ who values it as \emph{large}. 
To compensate, agent $j$ provides a suitable good in return, which can fall into two cases:
\begin{enumerate}
    \item the returned good is large for agent $i$, increasing her utility, or
    \item the returned good is small for both agents, maintaining agent $i$'s utility while improving agent $j$'s utility.
\end{enumerate}
\item[Type II:] ``one-many exchange'': This exchange occurs only under our personalized bi-valued setting, where agents have different \emph{ratios} of their large to small valuations.
Suppose two agents, $i$ and $j$, satisfy $r_i > r_j$. 
Agent $i$ gives $r_j$ goods, each valued small by both agents, to agent $j$. 
In return, agent $j$ provides a single good valued large by both agents. 
Here, agent $i$ increases her utility by $a_i-b_i\cdot r_j$, while agent $j$'s utility remains unchanged.

Intuitively, Type II exchanges exploit the differences in the agents' valuation \emph{ratios}—agent $i$ benefits by receiving one highly-valued good in exchange for losing several low-valued goods.
\end{itemize}

In both exchanges above, an agent gives an item to another agent who values it as ``large''. In the case that such an item does not exist, the above two types of exchanges can be naturally extended such that the large item is passed from one agent to another through some \emph{intermediate agents}. 
We call this \emph{cyclic extensions} or \emph{indirect extensions} of the two types of exchanges.
\begin{itemize}
    \item For Type I, if the recipient agent $j$ lacks a suitable good to compensate agent $i$, she may do so indirectly via intermediate agents. 
    Specifically, $j$ gives a good (valued large by both) to agent $a_1$, who then passes a similarly valued good to $a_2$, and so on, until agent $a_k$ transfers a suitable item—falling into one of the two cases—to agent $i$. 
    This forms a cycle through which agent $i$ ultimately receives compensation.

    \item For Type II, similarly, the large good can reach agent $i$ indirectly via intermediate agents, forming a cycle. 
    This is necessary if agent $j$ lacks a good that is valued highly by agent $i$.

\end{itemize}

Our characterization for Pareto-optimal allocations with integer $r_i$ shows that the two types of natural Pareto-improvements—along with their cyclic extensions—are sufficient to capture all violations of Pareto-optimality. 
That is, an allocation is not Pareto-optimal if and only if a Pareto-improvement of one of these basic types exists.
This result leads to a simple polynomial-time algorithm for verifying Pareto-optimality, as detecting such improvements reduces to finding cycles in a directed graph. Moreover, it also yields a polynomial-time method for computing a Pareto-optimal allocation that dominates a given input allocation.

Our second result shows that EFX allocations always exist for personalized bi-valued utilities.
Moreover, an EFX allocation can be computed in polynomial time (Sect.~\ref{sec:EFX}).
Our algorithm is an extension of the \emph{Match\&Freeze} algorithm by~\citet{amanatidis2021maximum} with the inclusion of an extra operation ``\emph{Modify}'' that incorporates the feature of personalized bi-valued utilities.
This result was also established independently in a recent work by \citet{byrka2025probingefxpmmsnonexistence}.

\subsection{Related Work}
There is a rich body of literature that deals with fairness and efficiency in the fair division field.
We first discuss research on the existence and computational complexity of allocations that are both (approximately) envy-free and Pareto optimal, followed by additional work relevant to our topic.

\paragraph{Existential results.}
We first present results relating to the existence of allocations that are both fair and efficient. 
As shown by~\cite{de2009complexity}, deciding Pareto-optimality under general additive utilities is coNP-complete, and determining whether there exists an allocation that is both envy-free and Pareto-optimal is $\Sigma_2^p$-complete.
For the general additive utilities, \citet{caragiannis2019unreasonable} proved that the allocations with the maximum Nash welfare are both EF1 and Pareto-optimal.
However, EFX and Pareto-optimality are in general incompatible.
Even for tri-valued valuations, \citet{garg2023computing} show that EFX and Pareto-optimal allocations may fail to exist, and deciding the existence of such allocations is NP-hard.

\paragraph{Computational results.}
Then we present results relating to efficient computations of allocations that are both fair and efficient. 
For the general additive utilities, \citet{barman2018finding} presented a pseudo-polynomial method to compute EF1 and PO allocations. 
\citet{garg2024computing} further presents a polynomial-time method to compute EF1 and fractionally~Pareto-optimal (a stronger notion of Pareto-optimality) allocations for \(k\)-ary valuations where each agent has at most \(k\) distinct values. 
As for restricted utility profiles, for binary valuations where the value of goods is either 0 or 1, \citet{halpern2020fair} and \citet{babaioff2021fair} independently show that we can simultaneously achieve EFX and Pareto-optimal via polynomial-time algorithms.
For bi-valued valuations, as mentioned before, EFX and Pareto-optimal allocations always exist (Maximum Nash welfare allocations)~\citep{amanatidis2021maximum} and can be computed in polynomial time~\citep{garg2023computing,bu2024best}.

\paragraph{Other related work.}
The research regarding fair and efficient allocations has also been done with alternative notions such as proportionality, maximin share, or social welfare maximization \citep{bei2012optimal,brams2012maxsum,aumann2015efficiency,bei2021price,li2024complete,barman2020optimal,aziz2020polynomial,aziz2023computing,nguyen2023fewtypes,feldman2024optimal,bu2025approxfairwelfare}. 
The compatibility of almost envy-freeness and (fractional) Pareto-optimality has also been studied in the setting with the chores (items with negative utilities) and mixed manna (items with both positive and negative utilities) settings with general or bi-valued valuations~\citep{camacho2023generalized,garg2023computing,aziz2023twochoretypes,garg2023chores,lin2025bivaluedchores}.
Finally, due to the restrictiveness of (personalized) bi-valued utilities, we can achieve fairness in more settings where fairness is typically harder to achieve, such as the online setting~\citep{amanatidis2025online2value,song2025onlinemmschores,wang2025online}, the setting with strategic agents~\citep{bu2024truthful,barman2019fair,halpern2020fair,babaioff2021fair}, etc.\section{Preliminaries}
\label{sec:prelim}
Given a positive integer $k$, let $[k]=\{1,\ldots,k\}$ and $[k]_0 = \{0, 1, \ldots, k\}$.

A fair division instance is denoted by $I = (N, M, \{u_i\}_{i\in [n]})$, where $N=[n]$ denotes the agents and $M = [m]$ denotes the indivisible goods (or items, interchangeably) to be allocated. 
Each agent $i \in [n]$ has an \emph{additive} utility function $u_i$, meaning that for every subset of goods $X_i \subseteq M$ allocated to agent $i$, the agent's utility is given by $u_i(X_i) = \sum_{g\in X_i}{u_i(\{g\})}$, where $u_i(\{g\})$ is agent $i$'s value on item $g$. 
We will write $u_i(g)$ instead of $u_i(\{g\})$ for simplicity. 
An allocation is a partition $\mathbf{X} = (X_1, \ldots, X_n)$ of $M$ amongst $n$ agents, where $X_i$ is the bundle allocated to agent $i$.
A partial allocation is a partition of a subset of $M$ amongst $n$ agents. 

In this paper, we focus on the \emph{personalized bi-valued} setting where for each agent $i$, there exist a pair of values $a_i,b_i$ with $a_i>b_i>0$ such that for any $g \in M$, $u_i(g)\in\{a_i, b_i\}$.
For all the theorems, lemmas, and propositions in this paper, valuations are assumed to be personalized bi-valued.
An item $g$ is \emph{large} (\emph{small}, resp.) for agent $i$ if $u_i(g)=a_i$ ($u_i(g)=b_i$, resp.).
We define the value \emph{ratio} of agent $i$ as $r_i = a_i/b_i$.

An allocation is \emph{Pareto-optimal} if there does not exist another allocation making some of the agents better off without making some others worse off. 
To be more specific, an allocation $\mathbf{X}$ is \emph{Pareto-optimal} if there does not exist another allocation $\mathbf{X'}$ satisfying that, 1) for each $i \in [n]$, we have $u_i(X'_i) \geq u_i(X_i)$, and 2) there exists $j \in [n]$ such that $u_j(X'_j) > u_j(X_j)$. 
If such a new allocation exists $\mathbf{X}'$, we say that the original allocation $\mathbf{X}$ is \emph{Pareto-dominated} by $\mathbf{X}'$ and the allocation $\mathbf{X'}$ is a \emph{Pareto-improvement} of $\mathbf{X}$.

For fairness, we use the widely accepted \emph{envy-freeness} concept and its relaxations. 
An allocation $\mathbf{X}$ is said to be: 
\begin{itemize}
    \item \emph{Envy-free (EF)} if for all $i, j \in N$, it satiesfies that $u_i (X_i) \geq u_i(X_j)$
    \item \emph{Envy-free up to any good (EFX)} if for all $i, j \in N$, for any good $g \in X_j$, it satisfies that $u_i(X_i) \geq u_i(X_j\backslash \{g\})$.
\end{itemize}

Considering the key properties of allocations discussed above, we present a lemma demonstrating the equivalence of allocations under a transformation of utility profiles:

\begin{lemma}[Equivalence under Ratio-Based Utilities]
\label{lemma:equivalence-ratio}
    An allocation is EFX (or Pareto-optimal) under the personalized bi-valued utility profile $\{(a_i,b_i)\}_{i=1,\ldots,n}$ if and only if it is EFX (or Pareto-optimal) under the profile $\{(r_i,1)\}_{i=1,\ldots,n}$ where $r_i=a_i/b_i$.
\end{lemma}

\begin{proof}
    Let $\{v_i\}_{i=1,\ldots,n}$ denote the profile $\{(a_i,b_i)\}_{i=1,\ldots,n}$ and $\{u_i\}_{i=1,\ldots,n}$ denote the profile  $\{(r_i,1)\}_{i=1,\ldots,n}$. It is straightforward to see that $v_i(S)=b_i\cdot u_i(S)$ holds for any $S\subseteq M$ and any agent $i$.
    Therefore, the EFX criteria $v_i(X_i)\geq v_i(X_j\setminus\{g\})$ and $u_i(X_i)\geq u_i(X_j\setminus\{g\})$ are equivalent. 
    
    Similarly, the Pareto-optimal criteria--"there does not exist another allocation $X'$ such that for all $i$, $v_i(X'_i) \geq v_i(X_i)$ and exists some $j$, $v_j(X'_j) > v_j(X_j)$" and "there does not exist another allocation $X'$ such that for all $i$, $u_i(X'_i) \geq u_i(X_i)$ and exists some $j$, $u_j(X'_j) > u_j(X_j)$"--are equivalent. 

    Hence we substitute $\{u_i\}_{i\in [n]}$ with $\{r_i\}_{i\in [n]}$ when discussing the topics mentioned above. 
\end{proof}

In the remaining part of this paper, we will use $r_i$ and $1$ instead of $a_i$ and $b_i$ to denote the values of a large item and a small item for agent $i$.

\section{On Pareto-Optimal Allocations}
\label{sec:PO}
This section focuses on Pareto-optimal allocations. 
We will prove that the two types of Pareto-improvements, and their cyclic extensions, are sufficient to cover all possible Pareto-improvements, if each $r_i$ is an integer.
To formally state this characterization, we begin by introducing the \emph{minimum Pareto-improvement graph}, which captures the minimal item exchanges required for Pareto-improvements, and we establish several basic properties of this graph.

In Sect.~\ref{sec:PO_integer}, we present our main characterization result (Theorem~\ref{thm:main_po}) for the special case where each ratio $r_i$ is an integer. Specifically, we show that the minimum Pareto-improvement graph must take one of the two ``cyclic'' forms described in Sect.~\ref{sec:ourresults}. This characterization leads to a polynomial-time algorithm for verifying whether a given allocation is Pareto-optimal (Corollary~\ref{cor:main}).

The proof of Theorem~\ref{thm:main_po} is given in Sect.~\ref{sec:PO_integer_proof}. Finally, in Sect.~\ref{sec:coNP-complete}, we show that, when $r_i$ values are allowed to be fractional, verifying Pareto-optimality becomes coNP-complete.

We first define some graphs for future analysis. 
To begin with, we model the item exchange process during Pareto-improvements with the \emph{item-exchange graph}.

\begin{definition}[Item-Exchange Graph]
\label{def:Item-exchange Graph}
Given two allocations $\mathbf{X}^1$ and $\mathbf{X}^2$ of the same set of items $M$ among $n$ agents, the \emph{item-exchange graph} $H(\mathbf{X}^1, \mathbf{X}^2) = (N, W)$ is a multi-directed graph representing the item exchanges needed to transform $\mathbf{X}^1$ into $\mathbf{X}^2$.

The vertex set $N$ corresponds to the agents, and the edge set $W$ contains a directed edge $(i, j)$ for each item $w \in M$ that is held by agent $i$ in $\mathbf{X}^1$ and by agent $j$ in $\mathbf{X}^2$ (with $i \ne j$). That is,
\[
W = \{ (i, j) \mid w \in X^1_i \cap X^2_j,~ i \ne j \}.
\]
Multiple edges between the same pair of agents may exist if multiple items are transferred between them.

Each edge is labeled with the corresponding item and classified into one of four types, based on how the source and target agents value the item:
\begin{itemize}
    \item \textbf{LS} (Large-to-Small): the item is large for agent $i$ and small for agent $j$.
    \item \textbf{SS} (Small-to-Small): the item is small for both $i$ and $j$.
    \item \textbf{SL} (Small-to-Large): the item is small for $i$ and large for $j$.
    \item \textbf{LL} (Large-to-Large): the item is large for both $i$ and $j$.
\end{itemize}
\end{definition}

We then define the ``minimum'' Pareto-improvements of a Pareto-dominated allocation. 
\begin{definition}[Minimum Pareto-improvement Graph]
    Given a Pareto-dominated allocation $\mathbf{X}$, the \emph{minimum Pareto-improvement graph} $H_{\min}(\mathbf{X})$ is an item-exchange graph $H(\mathbf{X},\mathbf{X}')$ such that $\mathbf{X}'$ is a Pareto-improvement of $\mathbf{X}$ that minimizes the number of the edges in $H(\mathbf{X},\mathbf{X}')$.
\end{definition}

We then present several propositions concerning the properties of the \emph{minimum Pareto-improvement graph}.

In a directed graph, we call a node a \emph{source} if its in-degree is $0$, and we call a node a \emph{sink} if its out-degree is $0$.

\begin{proposition}
\label{prop:no-source-no-sink}
    Given any allocation $\mathbf{X}$, $H_{\min}(\mathbf{X})$ contains no source and sink node. 
\end{proposition}
\begin{proof}
In a Pareto-improvement, no agent’s utility can decrease, so source nodes cannot exist in $H_{\min}(\mathbf{X})$. 
All incoming edges to a sink node can be removed without affecting the improvement, contradicting the minimality of $H_{\min}(\mathbf{X})$.
\end{proof}

\begin{proposition}
\label{prop:no-LS-edge}
    Given any allocation $\mathbf{X}$, $H_{\min}(\mathbf{X})$ contains no LS (Large-to-Small) edges.
\end{proposition}
\begin{proof}
Suppose for contradiction that $H_{\min}(\mathbf{X})=H(\mathbf{X},\mathbf{X}')$ contains an LS edge corresponding to item $g$ transferred from agent $i$ (who values $g$ as large) to agent $j$ (who values it as small). In a Pareto-improvement, agent $i$ must be compensated by receiving another item $h$—--from some agent $j'$ (possibly $j$)—--to avoid a utility loss.

Now consider modifying the exchange: instead of transferring $g$ from $i$ to $j$ and $h$ from $j'$ to $i$, we transfer only $h$ from $j'$ to $j$, leaving $g$ with agent $i$. 
Let $\mathbf{X}''$ be the resultant new allocation.

We claim that $\mathbf{X}''$ is also a Pareto-improvement over $\mathbf{X}$. In $\mathbf{X}'$, agent $i$ loses a large item $g$ and gains $h$ (which may be small or large), whereas, in $\mathbf{X}''$, $i$ keeps $g$ and receives nothing—so their utility is weakly higher in $\mathbf{X}''$. For agent $j$, item $g$ (small to them) is replaced with $h$ (which may be small or large), so their utility is also weakly higher. All other agents are unaffected.

Since $\mathbf{X}''$ results in a Pareto-improvement with fewer edges than $\mathbf{X}'$, this contradicts the minimality of $H_{\min}(\mathbf{X})$.
\end{proof}

\begin{proposition}
\label{prop:no-SS-path}
    Given any allocation $\mathbf{X}$, $H_{\min}(\mathbf{X})$ contains no two consecutive SS (Small-to-Small) edges.
\end{proposition}

\begin{proof}
Suppose for contradiction that $H_{\min}(\mathbf{X})$ contains two consecutive SS edges: item $g_1$ transferred from agent $i_1$ to $i_2$, and item $g_2$ from $i_2$ to $i_3$, where both items are valued as small by their respective senders and receivers.

We construct a modified exchange by removing both edges and directly transferring item $g_1$ from $i_1$ to $i_3$, leaving $g_2$ with $i_2$. Since $g_1$ is small for both $i_1$ and $i_2$, and $g_2$ is small for both $i_2$ and $i_3$, this reassignment preserves or increases the utility of each involved agent. In particular, agent $i_3$ may value $g_1$ more than $g_2$, so their utilities do not decrease.

Thus, the resulting allocation remains a Pareto-improvement with fewer edges, contradicting the minimality of $H_{\min}(\mathbf{X})$.
\end{proof}

We finally define the \emph{large-item transfer possibility graph}, which captures potential large-item transfers between agents under a given allocation.

\begin{definition}[Large-Item Transfer Possibility Graph]
\label{def:LTP-graph}
Given an allocation $\mathbf{X}$, the \emph{large-item transfer possibility graph} is a directed graph $I(\mathbf{X}) = (V, E)$, where $V$ is the set of agents and
\[
E = \{ (i, j) \in V \times V : \exists g \in X_i \text{ such that } u_j(g) = r_j \}.
\]
\end{definition}

In $I(\mathbf{X})$, an edge $(i, j) \in E$ indicates that agent $i$ holds an item that agent $j$ values as large, meaning that it is possible for agent $i$ to transfer a large-valued (with respect to agent $j$) item to agent $j$.

\subsection{Special Case Where $r_i$'s Are Integers}
\label{sec:PO_integer}
As mentioned in Sect.~\ref{sec:ourresults}, under the personalized bi-valued setting, if $r_i$ is an integer for every agent $i$, Pareto-optimal allocations have a nice structure: for any Pareto-dominated allocation, there always exist two specific types of Pareto-improvements.
The intuition given in Sect.~\ref{sec:ourresults} is formally stated in the theorem below.

In this section, we will see that our characterization implies a polynomial-time algorithm to decide if an allocation is Pareto-optimal.
In Sect.~\ref{sec:PO_integer_proof}, we will prove Theorem~\ref{thm:main_po}.

\begin{theorem}
\label{thm:main_po}
Given a Pareto-dominated allocation $\mathbf{X}$ with integer-ratio utilities $\{r_i\}_{i \in [n]}$, the corresponding minimum Pareto-improvement graph $H_{\min}(\mathbf{X})$ must fall into one of the following two structural categories (where ``$\xrightarrow{\text{LL}} \cdots \xrightarrow{\text{LL}}$'' denotes a path consisting of one or more LL edges):
\begin{itemize}
    \item \textbf{Type I: Small-to-Large Exchange Cycle} \\
    $H_{\min}(\mathbf{X})$ forms a directed cycle:
    \[
    H_1: a_1 \xrightarrow{\text{SL}} a_2 \xrightarrow{\text{LL}} \cdots \xrightarrow{\text{LL}} a_k \xrightarrow{\text{LL/SL/SS}} a_1
    \]
    
    \item \textbf{Type II: One-to-Many Exchange Cycle} \\
    $H_{\min}(\mathbf{X})$ has the following structure:
    \[
    H_2: a_k \xrightarrow{r_1 \times \text{SS}} a_1 \xrightarrow{\text{LL}} \cdots \xrightarrow{\text{LL}} a_{k-1} \xrightarrow{\text{LL}} a_k
    \]
    where $r_1 < r_k$, and $\xrightarrow{r_1 \times \text{SS}}$ denotes $r_1$ SS edges from $a_k$ to $a_1$.
\end{itemize}
\end{theorem}

\noindent
Both structures described above represent item exchanges that constitute valid Pareto-improvements:

In \textbf{Type I}, the cycle ensures that one agent's utility strictly increases without harming others:
If the closing edge $(a_k, a_1)$ is LL, then $a_1$ strictly gains;
If it is SL or SS, then $a_k$ strictly gains.
    
In \textbf{Type II}, agent $a_k$ gives $r_1$ small-valued items to $a_1$, and receives a single large-valued item through the LL path from $a_{k-1}$. Since $r_k > r_1$, this exchange strictly increases $a_k$’s utility while leaving other agents' utilities unchanged.

With Theorem~\ref{thm:main_po}, we can then prove that it is possible to determine whether an allocation with integer-ratio utilities can be decided in polynomial time.

\begin{corollary}\label{cor:main}
    Given an allocation $\mathbf{X}$ with integer-ratio utilities $\{r_i\}_{i\in [n]}$, we can decide in polynomial time if $\mathbf{X}$ is Pareto-optimal.
    In addition, if $\mathbf{X}$ is not Pareto-optimal, we can find in polynomial time a Pareto-optimal allocation $\mathbf{X}'$ that Pareto-dominates $\mathbf{X}$.
\end{corollary}
\begin{proof}
    For a given allocation $\mathbf{X}$, construct a large-item transfer possibility graph $I(\mathbf{X}) = (V, E)$ according to Definition~\ref{def:LTP-graph}. 
    According to Theorem~\ref{thm:main_po}, $\mathbf{X}$ is Pareto-optimal if and only if $I(\mathbf{X})$ does not contain subgraphs $H_1$ or $H_2$. 

    To check the existence of \emph{Type I} improvements, consider each agent pair $(a_1, a_2)$ where $a_2$ views some item in $X_1$ as large while $a_1$ views it as small.
    Temporarily add an ``SL'' edge $(a_1, a_2)$ to $G$, and apply the $BFS$ method from $a_2$ to determine whether there exists a path of one or more consecutive $LL$ edges from $a_2$ leading to some agent, followed by a single $LL$, $SL$, or $SS$ edge returning to $a_1$.
    Remove the temporary edge after each search.
    If such a path exists, a Type I Pareto improvement is found; otherwise, none exists.

    Next, to check \emph{Type II} improvements, for each agent pair $(a_k, a_1)$ such that $X_k$ contains more than $r_1$ small items and $r_k > r_1$, temporarily add $r_1$ ``SS'' edges from $a_k$ to $a_1$ in $G$. 
    Then apply $BFS$ from $a_1$ to check whether a path of one or more consecutive $LL$ edges returns to $a_k$. 
    Remove the temporary edges after each search. 
    If such a path exists, a \emph{Type II} Pareto improvement is found; otherwise, none exists. 

    Since both constructing $G$ and checking the two types of Pareto improvements can be done in polynomial time, the entire verification process runs in polynomial time. 
\end{proof}

\subsection{Proof of Theorem~\ref{thm:main_po}}
\label{sec:PO_integer_proof}
We prove Theorem~\ref{thm:main_po} in this subsection.
Throughout this subsection, we fix an arbitrary Pareto-dominated allocation $\mathbf{X}$ and denote $H_{\min} = H_{\min}(\mathbf{X})$ for simplicity.

The proof proceeds by analyzing two cases, depending on whether $H_{\min}$ contains an SL edge:

\begin{enumerate}
    \item If $H_{\min}$ contains at least one SL (Small-to-Large) edge, we show that a Type I structure in Theorem~\ref{thm:main_po} must exist.
    \item If $H_{\min}$ contains no SL edges, we show that a Type II in Theorem~\ref{thm:main_po} structure must exist.
\end{enumerate}
In either case, the existence of the corresponding subgraph $H_1$ or $H_2$ necessarily implies that the entire graph $H_{\min}$ must be exactly $H_1$ or $H_2$, due to the minimality of $H_{\min}$.
Therefore, it suffices to prove 1 and 2 above, and we prove them in the next two subsections.

\subsubsection{$H_{\min}$ Contains an SL Edge} 
Suppose $H_{\min}$ contains an SL edge from agent $a$ to agent $b$. 
By Proposition~\ref{prop:no-source-no-sink}, $H_{\min}$ has neither source nor sink nodes. 
Therefore, starting from the edge $a \xrightarrow{\text{SL}} b$, there must exist a directed cycle in $H_{\min}$ that eventually returns to $a$. 
Moreover, since LS edges are forbidden by Proposition~\ref{prop:no-LS-edge}, all edges in this cycle must be of types LL, SL, or SS.

Let the cycle be $a\xrightarrow{SL} b\to c_1\to c_2\to\cdots\to c_k\to a$.
If all edges except for the first edge $(a,b)$ are of type LL, then the cycle is of Type I ($H_1$), as claimed. 
Now suppose there exists at least one non-LL edge along the path $b \rightarrow c_1 \rightarrow \cdots \rightarrow c_k \rightarrow a$. 
If the only non-LL edge is the final edge $(c_k, a)$, then the cycle still conforms to the $H_1$-type structure.

Otherwise, let $(c_i,c_{i+1})$ be the first non-LL edge in the path.
Since LS edges are not allowed, this edge must be either SL or SS.
Consider now the shortened cycle:
\[
a \xrightarrow{\text{SL}} b \rightarrow c_1 \rightarrow \cdots \rightarrow c_i \rightarrow a,
\]
where, for the item transfer corresponding to the last edge $(c_i,a)$, instead of letting agent $c_i$ pass an item to $c_{i+1}$, we let $c_i$ pass this item directly to $a$.

It is straightforward to check that this shorter cycle also gives a Pareto-improvement of Type I.
\begin{itemize}
\item Agent $c_i$ gives away an item she values as small and receives one she values as large, gaining utility.
    \item Agent $a$ gives away a small-valued item and receives an item that is either small or large, so her utility does not decrease.
    \item All other agents on the cycle continue to give and receive large items, leaving their utilities unchanged.
\end{itemize}

\subsubsection{$H_{\min}$ Contains No SL Edge}

We now analyze the case when $H_{\min}$ contains no SL edge. 
Recall from Proposition~\ref{prop:no-LS-edge} that $H_{\min}$ contains no LS edges either. Thus, each edge in $H_{\min}$ represents either an LL edge or an SS edge. 
Accordingly, throughout this subsection, when we describe an item in $H_{\min}$ as ``large'' or ``small'', we mean it is viewed as such by both agents involved in the transfer.

We classify the agents based on their role in $H_{\min}$ as follows:
\begin{itemize}
    \item A \emph{net large-item giver} has strictly more outgoing LL edges than incoming LL edges.
    \item A \emph{net large-item receiver} has strictly more incoming LL edges than outgoing LL edges.
\end{itemize}

By Proposition~\ref{prop:no-SS-path}, $H_{\min}$ contains no two consecutive SS edges, so no agent simultaneously gives and receives small items. 
Moreover, since $H_{\min}$ represents a Pareto-improvement, any agent giving away small items must be compensated by large items, and she must be a net large-item receiver. 
Conversely, due to minimality, an agent receiving small items must be a net large-item giver. 
Additionally, each agent who neither gives nor receives small items must give away as many large items as she receives (i.e., she is neither a net large-item giver nor a net large-item receiver); otherwise, the minimality of $H_{\min}$ or the Pareto-improvement property would be violated.

Thus, agents in $H_{\min}$ can be categorized exhaustively into four types:
\begin{enumerate}
    \item Agents who receive small items and are \emph{net large-item givers}.
    \item Agents who give small items and are \emph{net large-item receivers}.
    \item Agents who neither give nor receive small items and give away as many large items as they receive (``intermediate agents'').
    \item Agents who do not participate in exchanges.
\end{enumerate}

Notice that agents of the third and fourth types are inconsequential in our subsequent analysis, since agents of the third type merely serve as intermediaries.

Focusing on the ``net effect'' of item exchanges, we derive a natural bipartite structure defined as follows.

\begin{definition}[Item-Exchange Bipartite Graph]
Given a Pareto-dominated allocation $\mathbf{X}$ and its corresponding minimum Pareto-improvement graph $H_{\min}(\mathbf{X})$, the \emph{item-exchange bipartite graph} $B_{\min}(\mathbf{X}) = (L \cup S, E)$ is a multi-directed bipartite graph where:
\begin{itemize}
    \item $L = \{l_1, \ldots, l_p\}$ contains agents who are net large-item givers.
    \item $S = \{s_1, \ldots, s_q\}$ contains agents who are net large-item receivers.
    \item Edges in $E$ represent aggregated item transfers:
    \begin{itemize}
        \item For each small item transferred from $s_j \in S$ to $l_i \in L$ in $H_{\min}$, we add one directed edge $(s_j, l_i)$ in $E$.
        \item For large-item transfers, we create a working copy of $H_{\min}$ and iteratively identify paths of consecutive LL edges that begin at some $l_i \in L$ and end at some $s_j \in S$. For each such path, we add a directed edge $(l_i, s_j)$ to $E$ and remove the corresponding LL edges from the working copy. This process continues until no LL edges remain.
    \end{itemize}
\end{itemize}
\end{definition}

For simplicity, we denote this bipartite graph by $B_{\min}$.

We next present a key lemma.

\begin{lemma}\label{lemma:Bmin-to-category2}
    If there exist nodes $l_i \in L$, $s_j \in S$ in $B_{\min}$ with an edge $(l_i, s_j)$ satisfying:
    \begin{enumerate}
        \item $r_{l_i} < r_{s_j}$, and
        \item the number of small items given by $s_j$ to nodes in $L$ is at least $r_{l_i}$,
    \end{enumerate}
    then $H_{\min}$ contains a Type II ($H_2$) subgraph.
\end{lemma}

\begin{proof}
    By the definition of $B_{\min}$, an edge $(l_i, s_j)$ indicates a path of consecutive LL edges in $H_{\min}$ from agent $l_i$ to $s_j$. Thus, we immediately have an $H_2$ structure:
    \[
    H_2: s_j \xrightarrow{r_{l_i} \times \text{SS}} l_i \xrightarrow{\text{LL}} \cdots \xrightarrow{\text{LL}} s_j.
    \]
    Notice that instead of giving those $r_{l_i}$ items to possibly different agents in $L$, we just let agent $s_j$ directly give these items to $l_i$.
    This is precisely a Type II subgraph as defined in Theorem~\ref{thm:main_po}.
\end{proof}

We now show that an eligible $(l_i, s_j)$ pair described in Lemma~\ref{lemma:Bmin-to-category2} must exist in $B_{\min}$.

Let the net number of large items given by agents in $L$ be $\{A_1, \ldots, A_p\}$, and the number of small items they receive be $\{a_1, \ldots, a_p\}$. Similarly, let the net number of large items received by agents in $S$ be $\{C_1, \ldots, C_q\}$, and the number of small items they give be $\{c_1, \ldots, c_q\}$.

By the Pareto-improvement property, each agent’s utility must weakly increase. For agents in $L$, this implies
\[
\sum_{i=1}^{p} A_i \cdot r_{l_i} \;\leq\; \sum_{i=1}^{p} a_i \;=\; \sum_{j=1}^{q} c_j.
\]

We define the \emph{surplus function} 
\[
\mathsf{surplus}: S \rightarrow \mathbb{R}
\]
measuring the excess value of small items given by agent \( s_j \), after accounting for the total incoming value (according to the ratios of senders) of large items she receives.
That is, for each agent \( s_j \in S \),
\[
\mathsf{surplus}(s_j) = c_j - \sum_{\substack{l_i \in L \\ (l_i, s_j) \in E}} w_{ij} \cdot r_{l_i},
\]
where \( w_{ij} \in \mathbb{N} \) denotes the number of large-item edges from agent \( l_i \) to agent \( s_j \) in $B_{\min}$.

Summing over $\mathsf{surplus}(s_j)$ for all \( s_j \in S \), we obtain:
\[
\sum_{j=1}^q \mathsf{surplus}(s_j) = \sum_{j=1}^q c_j - \sum_{i=1}^p A_i \cdot r_{l_i} \;\geq\; 0.
\]

Therefore, there exists at least one agent \( s_j \in S \) such that \( \mathsf{surplus}(s_j) \geq 0 \).
Recalling that $s_j$ is a net large-item receiver (so the summation in the second term of $\mathsf{surplus}(s_j)$ contains at least one term), this implies that there exists an edge $(l_i,s_j)$ in $B_{\min}$ such that the number of small items given away by agent $s_j$ is at least $r_{l_i}$, i.e., $c_j\geq r_{l_i}$.
If we further have $r_{l_i}<r_{s_j}$, $H_{\min}$ contains $H_2$ by Lemma~\ref{lemma:Bmin-to-category2}.
If $r_{l_i}\geq r_{s_j}$, we derive a contradiction by showing that the corresponding $H_{\min}(\mathbf{X})$ cannot be minimum.
Indeed, we can remove the path corresponding to $(l_i,s_j)$ in $H_{\min}(\mathbf{X})$; agent $l_i$ receives $r_{l_i}$ less small items than before (notice that agent $l_i$ receives at least $r_{l_i}$ small items before, in order to ensure Pareto-optimality); agent $s_j$ gives away $c_i-r_{l_i}$ small items (instead of $c_i$ items as before).
Since both the total number of small items given away by the agents in $S$ and the total number of small items taken by agents in $L$ are reduced by the same number $r_{l_i}$, this ensures the utilities for all agents but $s_j$ remain unchanged, and the utility of agent $s_j$ increases.
In addition, fewer edges are introduced in this new Pareto-improvement.
Thus, we have a contradiction.

This completes the proof of Theorem~\ref{thm:main_po}.

\subsection{General Case Where $r_i$ May Be Fractional}
\label{sec:coNP-complete}
We next investigate the complexity of determining Pareto-optimality in general personalized bi-valued instances with fractional ratios. 
We show that the problem is coNP-complete.

\begin{theorem}\label{thm:coNP-completeness}
    Given an allocation $\mathbf{X}$ with utility profiles $\{(r_i,1)\}_{i \in [n]}$, deciding the Pareto-optimality of $\mathbf{X}$ is $\text{coNP-complete}$ even when $r_i \in \{r_a, r_b\}$. 
\end{theorem}

To give some intuitions, if the ratios of the two agents are close, then we may need ``many-many exchange'' instead of the ``one-many exchange'' described in Type II.
For example, for two agents $i$ and $j$ with ratios $r_i=\frac43$ and $r_j=\frac43+\varepsilon$, a Pareto-improvement may be that agent $i$ gives $3$ large items to agent $j$, in exchange of $4$ small items from agent $j$.
This type of Pareto-improvements can no longer be easily checked by finding cycles in a directed graph.
In addition, more complicated scenarios can happen where a set of $T$ agents with ratios $\frac43$ give away a total number of $3T$ items to a set of $T$ agents with ratios slightly larger than $\frac43$, in exchange for a total of $4T$ small items.
This requires that the large-item transfer possibility graph $I(\mathbf{X})$ contains a 3-regular $T\times T$ bipartite subgraph, and deciding the existence of a regular bipartite subgraph from a bipartite graph is known to be NP-complete~\citep{PLESNIK1984161}.

\paragraph{Proof of Theorem~\ref{thm:coNP-completeness}.}
We present a reduction from the following NP-complete problem established by~\citet{PLESNIK1984161}.

\begin{definition}[$k$-Regular Bipartite Subgraph Problem]
Let $k \geq 3$ be a fixed integer.  
Given a bipartite graph $G = (A \cup B, E)$ in which each vertex has degree at most $k+1$, decide whether $G$ contains a $k$-regular bipartite subgraph.
\end{definition}

Given an instance $G = (A \cup B, E)$ of the $k$-Regular Bipartite Subgraph Problem, we construct a corresponding instance of the allocation problem under personalized bi-valued utilities.

\begin{instance}
\label{ins:kregular}
    We construct an instance with $n = |A \cup B|$ agents and $m = |E| + (k+1)|B|$ items.
    We begin by constructing one agent for each vertex in $A\cup B$.
    For each edge $(a,b)\in E$ with $a\in A$ and $b\in B$, we place an item in agent $a$'s bundle. 
    This item is considered a small item by all agents except agents a and b, for whom it is a large item. 
    Additionally, for each agent in $B$, we construct $k+1$ items that are small for all agents, and assign them to this agent's bundle.
    This completes the construction of the allocation. 
    Observe that the bundle of each agent \( a \in A \) contains at most \( k+1 \) large items, as the degree of any vertex in \( G \) is at most \( k+1 \). 
    These items may be perceived as either large or small by agents in \( B \), depending on the graph structure. 
    On the other hand, each agent \( b \in B \) receives a bundle consisting precisely of \( k+1 \) small items.

    The utility profile is defined as follows: for each agent \( a \in A \), the utility values are drawn from \( \{(r_a = 1 + 1/k, 1)\} \); and for each agent \( b \in B \), the values are from \( \{(r_b = 1 + \frac{1+\epsilon}{k}, 1)\} \), where \( \epsilon \) is a sufficiently small positive real number satisfying \( \epsilon < 1/(km) \).
\end{instance}

It is easy to see that the large-item transfer possibility graph $I(\mathbf{X})$ (Definition~\ref{def:LTP-graph}) corresponding to the allocation $\mathbf{X}$ in Instance~\ref{ins:kregular} coincides exactly with the graph $G$, with the directed edges from $A$ to $B$ in $I(\mathbf{X})$ viewed as being undirected.

The following lemma concludes Theorem~\ref{thm:coNP-completeness}.

\begin{lemma}
\label{lemma:generalpo}
Let $\mathbf{X}$ be the allocation described in Instance~\ref{ins:kregular}.
$G$ contains a $k$-regular subgraph if and only if $\mathbf{X}$ is Pareto-dominated. 
\end{lemma}
\begin{proof}
    Assume the vertices in the $k$-regular subgraph is $A_1 \cup B_1$ with $A_1 \subseteq A$ and $B_1\subseteq B$. 
    By regularity, we have $|A_1| = |B_1|$. 
    
    [Sufficiency]
    This follows immediately: each agent in $A_1$ gives up $k$ large items (following the edges in the $k$-regular subgraph) and receives $k+1$ small items, while each agent in $B$ exchanges $k+1$ small items for $k$ large ones. Under this arrangement, each $A_1$-agent retains the same utility, and every $B_1$-agent's utility is increased by $\epsilon$.
    
    [Necessity]
    We prove that, if the allocation $\mathbf{X}$ is Pareto-dominated, all ``LL'' edges in the corresponding $H_{\min}(\mathbf{X})$ must form a directed $k$-regular subgraph which corresponds to an (undirected) $k$-regular subgraph of $G$. 

    Let $U(\mathbf{X}) = \sum_{i \in A \cup B}{u_i(X_i)}$ be the social welfare of $\mathbf{X}$. 
    In Instance \ref{ins:kregular}, $U(\mathbf{X}) = |E|\cdot r_a + |B|\cdot (k+1)$ with all large items allocated to agents in $A$. 
    Assume that $\mathbf{X'}$ is a Pareto-improvement of $\mathbf{X}$, it is easy to see that $U(\mathbf{X'}) \leq U_{max} = |E|\cdot r_b + |B|\cdot (k+1)$, where $U_{max}$ is the maximum possible social welfare where all large items are allocated to agents in $B$ with a larger large-to-small ratio. 
    Therefore, the maximum increment of social welfare of $\mathbf{X}$ is $\max(\Delta(U)) = U_{max} - U(\mathbf{X}) = |E|\cdot \epsilon<1/k$.
    This implies that
    \begin{center}
        (*) every agent's utility increment from $\mathbf{X}$ to $\mathbf{X}'$ is strictly less than $1/k$.
    \end{center}
    Moreover, as at least one of the agents' utility improves while others do not decrease, $U$ strictly increases during a Pareto-improvement. 
    Thus, $H_{\min}$ must include some LL edges from $A$ to $B$. 

    Let $a$ be an arbitrary agent in $A$ who gives out large items according to $H_{\min}$.
    Assume that $a$ gives out $x$ large items and receives $y$ small ones in return.
    Notice that $a$'s utility is changed by $\Delta(u_a) = y - x(1+1/k) = y - x - x/k$. 
    In the following, we prove by contradiction that we must have $x=k$.
    \begin{itemize}
        \item $0 < x < k$\\
        Then $a$ should receive at least $x+1$ small items from agents in $B$ for her loss. 
        Therefore, with $y \geq x+1$, the utility gain for agent $a$ is $\Delta(u_a) =y - x - x/k \geq 1 - x/k\geq 1/k$, which contradicts to (*).
        \item $x>k$\\
        Since by our construction that each agent in $A$ receives at most $k+1$ items, we have $x=k+1$.
        To ensure that the utility of $a$ does not decrease, we need that $\Delta(u_a)=y-(k+1)-(k+1)/k\geq0$, which implies $y\geq k+3$.
        However, if this is the case, we have $\Delta(u_a)\geq (k+3)-(k+1)-(k+1)/k=1-1/k>1/k$, which again contradicts to (*).
    \end{itemize}
    Given $x=k$, we must also have $y=k+1$.
    Firstly, to make sure agent $a$'s utility does not decrease, we must first have $y\geq k+1$.
    If $y=k+1$, the utility of agent $a$ remains unchanged: $\Delta(u_a)=(k+1)-k-k/k=0$.
    Therefore, if $y>k+1$, agent $a$'s utility is increased by at least $1$, which contradicts (*).

    Therefore, each agent in $A$ either does not participate in the item exchange process of the Pareto-improvement or gives out $k$ large items and receives $k+1$ small ones. 

    With a similar analysis based on the bound $1/k$ on utility gains in (*), we can prove that all agents in $B$ either do not participate in the item exchange process of the Pareto-improvement or give out $k+1$ small items and receive $k$ large ones.

    Therefore, $H_{\min}(\mathbf{X})$ forms a $k$-regular subgraph, and it corresponds to a $k$-regular subgraph of $G$.
\end{proof}

\section{On EFX Allocations}
\label{sec:EFX}

We show that an EFX allocation always exists for personalized bi-valued utility instances, and it can be computed in polynomial time.
Our algorithm, named \emph{Match\&Modify\&Freeze}, builds upon the earlier ``match-and-freeze'' approach proposed by~\citet{amanatidis2021maximum} for simpler bi-valued valuations, where all agents share the same pair of valuation numbers. 
Our finding is concurrent with that of \citet{byrka2025probingefxpmmsnonexistence}, who independently reached the same conclusion through a similar modification.

\subsection{Technical Overview}
Our main technical contribution is the introduction of an additional \emph{modify} step specifically designed to handle the personalized nature of the agents' valuations. At a high level, the algorithm operates as follows:

\begin{enumerate}
    \item \textbf{Matching step:} In each round, the algorithm first finds a maximum matching between agents and available goods, considering edges that represent large-value items for each agent. 
    
    \item \textbf{Modify step (our contribution):} If the matching from the previous step is not perfect, meaning some agents are left unmatched, we carefully adjust it to prioritize agents with larger valuation ratios ($r_i = a_i / b_i$). This step ensures that agents who place a higher relative value on large items are allocated them preferentially, preserving fairness and avoiding unnecessary envy among agents. This careful adjustment distinguishes our method from previous approaches that consider uniform bi-valued utilities.
    
    \item \textbf{Freezing step:} After allocating items according to the matching, the algorithm ``freezes'' certain agents temporarily if they have received a large-valued good that might lead to envy from others. Frozen agents do not receive further goods for a few rounds, allowing the algorithm to manage potential envy by giving smaller-valued items to the remaining unfrozen agents.
\end{enumerate}

By repeatedly applying these steps, our method guarantees that the partial allocation after each round remains EFX—no agent strongly envies another agent even if an arbitrary item is removed from the other's bundle. When the algorithm terminates (when all goods are allocated), it produces an EFX allocation in polynomial time.

\subsection{Match\&Modify\&Freeze Method} 
We first provide a brief description of our method in Algorithm~\ref{alg:mf} and then define some necessary graph notions to describe the algorithm, and finally formally prove its correctness. 
Our method, named \emph{Match\&Modify\&Freeze}, follows the ideas of \citet{amanatidis2021maximum}, which discusses the more special setting with bi-valued valuations (where $a_i$ and $b_i$ are the same for all agents). 
Our method proceeds in \emph{rounds} (the loops in Line \ref{line:loop} of Algorithm \ref{alg:mf})  and maintains a set of \emph{unfrozen} agents. 
We begin by constructing a bipartite graph \(G = (V_1, V_2, E)\), where \(V_1\) represents the \emph{unfrozen} agents (at first all \(n\) agents) and \(V_2\) represents the remaining goods (at first all \(m\) goods). 
An edge connects agent \(i\) to good \(g\) precisely when \(g\) is a \emph{large} good to $i$. 

In each round, we allocate exactly one good to every \emph{unfrozen} agent, assuming sufficient goods remain available.
The process begins by finding a maximum matching between the agents and the goods in \(G\).
If the matching is not perfect (i.e., not every agent in \(V_1\) is matched with a good in \(V_2\)), we adjust it to prioritize agents with larger ratios while preserving the size of the matching. 
These adjustments are performed using Algorithm~\ref{alg:adjustmatching}, which involves additional techniques compared to the method in \citet{amanatidis2021maximum}. 

Once the adjustment for the matching is done, the matched goods are allocated to the corresponding agents. 
If the matching is perfect, each agent receives one large good, and we proceed to the next round.  
Otherwise, we arbitrarily allocate a remaining good to each unmatched agent, and the current round ends. 
The allocated goods are removed from \(V_2\) at the end of each round. 

Agents with larger ratios who are prioritized to receive large goods in an imperfect matching may be envied by other agents who miss out, either in the current round or in subsequent ones. 
To address this, we use a \emph{freezing} mechanism in Algorithm \ref{alg:freeze} that removes these agents from \(V_1\) for several rounds, and they will be \emph{unfrozen} afterward with the possible envy eliminated. 
At the end of each round, the graph \(G\) updates by removing the \emph{frozen} agents and reintroducing them when they are \emph{unfrozen} in \(V_1\). 

This iterative process continues until either all goods are allocated or the remaining goods become fewer than the number of \emph{unfrozen} agents in the current round.
In the latter case, the remaining goods are allocated according to the agents’ freezing times, giving priority to those who have never been frozen or who were frozen more recently.

Crucially, we ensure that after each round, the partial allocation constructed so far remains EFX.

\begin{theorem}\label{thm:main}
    For any personalized bi-valued instance $(N, M, \{r_i\}_{i\in [n]})$, Algorithm~\ref{alg:mf} computes an EFX allocation in polynomial time.
\end{theorem}

\begin{proof}
    See Section~\ref{sec:proof_thm_main}. 
\end{proof}

\subsection{Preliminaries}

Before proving Theorem~\ref{thm:main}, we present some definitions and observations related to the \\ \emph{Match\&Modify\&Freeze} method. 

We first define some graph notations related to the allocation process. 

\begin{definition}[Matching]
\label{def:matching}
    Given a bipartite graph $G = (V_1, V_2, E)$, a \emph{matching} $T = (V_1', V_2', E')$ is a subgraph of $G$ with $V_1' \in V_1$, $V_2' \in V_2$ and $E' = \{(i, g): i \in V_1', g \in V_2'\}$. 
\end{definition}

\begin{definition}[Alternating Path]
\label{def:alternating_path}
    Given a bipartite graph $G = (V_1, V_2, E)$ and a matching $T$, an \emph{alternating path} $P$ is a path in $G$ starting from an unmatched node with the edges alternating between those not in the matching $T$ and those in the matching $T$. 
    
    Let $P = \{i_0, (g_1, i_1), \ldots, (g_p, i_p)\}$ denotes the alternating path, consisting of nodes $\{i_0, g_1, i_1, \ldots, i_p, g_p\}$ with $\{i_x\}_{x \in [p]_0} \subseteq V_1$ and $\{g_y\}_{y \in [p]} \subseteq V_2$. 
    The edges in $P$ are $\{(i_0, g_1), (g_1, i_1), (i_1, g_2), \ldots,(i_{p-1}, g_p), (g_p, i_p)\}$, with $(g_1, i_1), (g_2, i_2), \ldots, (g_p, i_p)$ in the matching $T$ and the others not. 
\end{definition}

\begin{definition}[Augmenting Path]
\label{def:augmenting_path}
    Given a bipartite graph $G = (V_1, V_2, E)$ and a matching $T$, an \emph{augmenting path} is a special \emph{alternating path} that starts and ends in unmatched nodes. 
\end{definition}

We then present some properties of Algorithm \ref{alg:freeze}, which determines which agents to freeze.

\begin{proposition}\label{prop:frozen_nolarge}
    If an agent is frozen in some round, then she considers none of the remaining unallocated goods as large ones.
\end{proposition}

\begin{proof}
    We prove the proposition by contradiction. 
    Suppose that an agent $i_x$ is frozen in some round. 
    According to Algorithm~\ref{alg:freeze}, this indicates that agent $i_x$ is part of an alternating path, one end of which is not matched with any good. 
    We denote the alternating path $P = \{i_0, (g_1, i_1), \ldots, (g_p, i_p)\}$ with $i_0$ unmatched and $x \in [p]_0$ by Definition~\ref{def:alternating_path}. 
    
    If there exists one unallocated good $g'_x$ that agent $i_x$ considers large, there can be a matching with larger size with matching pairs being $\{(i_0, g_1), (i_1, g_2), \ldots,  (i_{x-1}, g_x), (i_x, g'_x), (i_{x+1}, g_{x+1}), \ldots, (i_p, g_p)\}$. 
    This contradicts the fact that the matching is maximum.
\end{proof}

\begin{corollary}\label{coro:frozenOnce}
    An agent will be frozen at most once throughout the entire allocation process.
\end{corollary}

\begin{proof}
Our algorithm only freezes an agent right after she receives a large good. 
This corollary then immediately follows from Proposition~\ref{prop:frozen_nolarge}.
\end{proof}

\begin{proposition}\label{prop:once_small_all_small}
    Once an agent is allocated a small good, she will never be allocated a large one in the remaining allocation process.
\end{proposition}
\begin{proof}
    Suppose that an agent $i$ is allocated a small good in some round.
    This indicates that she considers none of the remaining unallocated goods to be large; otherwise, the matching found in Algorithm~\ref{alg:mf} is not maximum.
\end{proof}

\begin{proposition}
\label{prop:frozenRounds}
    If an agent $i$ is set to be frozen for $s_i$ rounds by Algorithm \ref{alg:freeze}, then $s_i \leq r_i - 1$. 
\end{proposition}

\begin{proof}
    According to Algorithm \ref{alg:freeze}, the number of freezing rounds 
    $$s_i = max\{\lfloor r_{i_0}\rfloor, i_0 \text{and\ $i$ are on an alternating path}\}$$
    for agent $i$. 
    As we adjusted the matched pairs in Algorithm \ref{alg:adjustmatching} that only the agents with larger ratios on the alternating path can be matched, we have $r_i \geq r_{i_0}$ and therefore $s_i = max\{\lfloor r_{i_0} - 1\rfloor\} \leq \lfloor r_{i}\rfloor - 1 \leq r_i - 1$.
\end{proof}

Now we are ready to prove Theorem~\ref{thm:main}.

\subsection{Proof of Theorem~\ref{thm:main}}
\label{sec:proof_thm_main}
Our proof of the theorem employs a constructive approach, showing that the partial allocations after each \emph{round} are EFX. 

The allocation process in Algorithm~\ref{alg:mf} is divided into two phases:  
\begin{enumerate}
    \item While each maximum matching \(T\) is perfect, and 
    \item From the first time \(T\) is not perfect onward.
\end{enumerate}

\subsubsection{Partial Allocation during Phase 1}
During phase 1, the partial allocations always satisfy EFX since all agents exclusively obtain large goods sequentially, one for each in every round.
If all goods are allocated within this phase, the resulting complete allocation will be EFX.
Even if the final round's matching is imperfect, agents who do not receive goods obtain only one good less than their companions.
No one will envy another agent eliminating any one of the goods of another one's bundle.

\subsubsection{Partial Allocation during Phase 2}
During phase 2, assume that after the $(R-1)_{th}$ round, the partial allocation is EFX. 
Now consider the $R_{th}$ round. 
In this phase, the agents can be categorized into 3 groups based on their status at the end of the round: 

The \emph{frozen} agents who are frozen in the current or the previous round, exemplified by agent(s) $i$; 

The \emph{matched agents} who are matched with large goods but not frozen in the current round, exemplified by agent(s) $j$; 

And the \emph{unmatched agents} who are allocated small goods, exemplified by agent(s) $k$. 

\paragraph{Frozen agent \(i\) does not strongly envy others.}
Suppose agent \(i\) is frozen in round \(R' \le R\). By Proposition~\ref{prop:frozenRounds}, the number of rounds \(i\) spent frozen is strictly less than the utility \(i\) gained in round \(R'\). 
Furthermore, by Proposition~\ref{prop:frozen_nolarge}, no remaining goods are considered large from \(i\)’s perspective, so no agent can gain more than \(r_i\) while \(i\) is frozen. 
Finally, Algorithm \ref{alg:freeze} also freezes any agent who acquires a large item from \(i\)’s point of view in round \(R'\).
As a result, \(i\) cannot strongly envy any other agent.

\paragraph{Matched agent $j$ does not strongly envy others.}
According to Proposition~\ref{prop:once_small_all_small}, agent $j$ perceives herself as allocated with $R$ large goods.
Furthermore, she considers that all other agents are allocated at most $R$ goods, whether large or small. 
Consequently, she does not envy any other agent.

\paragraph{Unmatched agent $k$ does not strongly envy others.}

Case 1: Agent $k$ sees no large goods in round $R$

If none of the goods matched in round $R$ appear large to agent $k$, then every other agent receives either a small good or nothing. 
Since $k$ also obtains a small good, $k$ cannot strongly envy anyone after this round.

Case 2: Agent $k$ sees large goods in round $R$

If agent \(k\) views some matched good as large (allocated to, say, agent \(i'\)) in round \(R\), it follows that \(k\) must have received \(R-1\) large items in prior rounds (or else the previous matching would not have been maximum). 
Under Algorithm \ref{alg:freeze}, any agent who receives a large item by \(k\)’s standard is frozen for at least \(\lfloor r_k - 1\rfloor\) subsequent rounds, during which \(k\) acquires at least \(\lfloor r_k - 1\rfloor\) small items. 
This ensures the utility gap remains below 1, so \(k\) will not strongly envy \(i'\) once \(i'\) is unfrozen—even if they both receive additional items later.
Moreover, if there are insufficient goods for a complete final round, \(k\) takes priority over agent $i'$ since $i'$ is frozen earlier, so \(k\) still will not strongly envy \(i'\).

\begin{algorithm}[tb]
    \caption{Match\&Modify\&Freeze$(N, M, \{r_i\}_{i\in [n]})$}
    \label{alg:mf}
    \begin{algorithmic}[1] 
        \STATE  \textbf{Input}: the agents $N$, the goods $M$ and the utilities $\{r_i\}_{i \in [n]}$
        \STATE $V_1 \leftarrow N$, $V_2 \leftarrow M$
        \WHILE{$V_2 \neq \emptyset$} \label{line:loop}
        \STATE Construct the bipartite graph $G = (V_1, V_2, E)$
        \STATE Compute a maximum matching $T$ on $G$,
        \STATE Adjust the matching if it is not perfect with Algorithm \ref{alg:adjustmatching}:\\
        T = \emph{AdjustMatching(G, T)}
        \FOR{each matched pair $(i, g)$}
            \STATE Allocate $g$ to agent $i$
            \STATE Remove $g$ from $V_2$
        \ENDFOR
        \FOR{each unmatched agent $j$}
            \STATE Allocate one arbitrary unallocated good $g'$ to $j$
            \STATE Remove $g'$ from $V_2$
        \ENDFOR
        \STATE Freeze the agents in $V_1$ for $\{s_i|i \in [n]\} = FreezeAgents(G, T)$ rounds with Algorithm \ref{alg:freeze}. 
        \ENDWHILE
        \STATE \textbf{return} the allocation $\mathbf{X}$
    \end{algorithmic}
\end{algorithm}

\begin{algorithm}[tb]
    \caption{AdjustMatching$(G, T)$}
    \label{alg:adjustmatching}
    \begin{algorithmic}[1]
    \STATE \textbf{Input:} $\text{A bipartite graph } G = (V_1, V_2, E), \text{a matching } T$
    \FORALL{alternating path $P \in G$}
    \STATE Assume that $P = \{i_0, (g_1, i_1), \ldots, (g_p, i_p)\}$ by Definition~\ref{def:alternating_path}
    \IF{the ratio of agent $i_0$ is larger than that of $i_p$}
    \STATE Adjust the matched edges $\{(g_1, i_1), (g_2, i_2), \ldots, (g_p, i_p)\}$ in $T$ to $\{(i_0, g_1), \ldots, (i_{p-1}, g_p)\}$. 
    \ENDIF
    \ENDFOR
    \RETURN The modified matching $T$
    \end{algorithmic}
\end{algorithm}

\begin{algorithm}
    \caption{FreezeAgents$(G, T)$}
    \label{alg:freeze}
    \begin{algorithmic}[1]
    \STATE \textbf{Input:} $\text{A bipartite graph } G = (V_1, V_2, E), \text{a matching } T$
    \STATE Initialize $s_i \leftarrow 0\ for\ i \in [n]$
    \FORALL{alternating path $P \in G$}
        \STATE Assume that $P = \{i_0, (g_1, i_1), \ldots, (g_p, i_p)\}$ by Definition~\ref{def:alternating_path}
        \FORALL{matched agent $i_x \in \{i_1, \ldots, i_p\}$}
            \STATE $s_{i_x} \leftarrow max\{s_{i_x}, \lfloor r_{i_0}-1\rfloor\}$
        \ENDFOR
    \ENDFOR
    \RETURN The number of frozen rounds of all agents $\{s_i|i \in [n]\}$
    \end{algorithmic}
\end{algorithm}

\section{Discussions and Future Work}
In this paper, we have studied Pareto-optimal allocations and showed that EFX allocations always exist under the personalized bi-valued setting.
It is a natural future direction to study if an almost envy-free and Pareto-optimal allocation always exists under the same setting, and if it can be found in polynomial time in the case such an allocation always exists.

For the combination of EF1 and Pareto-optimality, we know that the maximum Nash welfare solution is always EF1 and Pareto-optimal even for general additive valuations~\citep{caragiannis2019unreasonable}.
However, designing a polynomial-time algorithm to find such an allocation is one of the central open problems in the fair division literature.
For personalized bi-valued utilities, the result about $k$-ary instances from \citet{garg2024computing} indicates that an EF1 and Pareto-optimal allocation can be found in polynomial time.

For general additive valuations, it is widely known that EFX is not always compatible with Pareto-optimality.
This is true even for tri-valued instances~\citep{garg2023computing}.
Given that EFX is compatible with Pareto-optimality for bi-valued instances~\citep{amanatidis2021maximum,garg2023computing,bu2024best}, we believe it is an interesting open problem to see \emph{if EFX is always compatible with Pareto-optimality for personalized bi-valued instances}.

Unfortunately, known techniques seem to fail to resolve this problem.
As we know, a maximum Nash welfare solution is both EFX and Pareto-optimal for bi-valued instances~\citep{amanatidis2021maximum}. One natural attempt is to see if it also works for personalized bi-valued instances. The following counterexample shows that a maximum Nash welfare solution may fail to be EFX (the maximum Nash welfare allocation gives agent 1 item 2 and agent 2 the remaining three items).

\begin{center}
    \begin{tabular}{|l|c|c|c|c|}
    \hline
   item & 1 & 2 & 3 & 4\\
   \hline
   agent 1  & 50 & 50 & 1 & 1 \\
   \hline
   agent 2  & 3 & 1 & 1 & 1 \\
   \hline
\end{tabular}
\end{center}

Another natural attempt is to apply the technique of \emph{Fisher market}, which has been successfully applied for finding fair and efficient allocations in many other settings~\citep{barman2018finding,garg2023computing,bu2024best}.
However, this approach is unlikely to work in the personalized bi-valued setting.
The efficiency guarantee obtained from this approach is the stronger notion of \emph{fractional Pareto-optimality} due to the \emph{first welfare theorem}, where an allocation is fractionally Pareto-optimal if it is not Pareto-dominated by any other \emph{fractional} allocation.
The following counterexample shows that fractional Pareto-optimality is not always compatible with EFX in the personalized bi-valued setting.
In this example, the only possible EFX allocations should give each agent exactly one large item from $\{1,2\}$ and exactly one small item from $\{3,4\}$.
Assume without loss of generality that agent $1$ receives $\{1,3\}$ and agent $2$ receives $\{2,4\}$.
This allocation is not fractionally Pareto-optimal: agent $1$ can give item $3$ to agent $2$, and in exchange for this, agent $2$ gives $1/3$ fraction of item $2$ to agent $1$.

\begin{center}
    \begin{tabular}{|l|c|c|c|c|}
    \hline
   item & 1 & 2 & 3 & 4\\
   \hline
   agent 1  & 6 & 6 & 1 & 1 \\
   \hline
   agent 2  & 3 & 3 & 1 & 1 \\
   \hline
\end{tabular}
\end{center}

Thus, resolving the compatibility of EFX and Pareto-optimality for the personalized bi-valued setting requires new techniques, and our characterization for Pareto-optimal allocations may be potentially useful.

\bibliography{ref}
\end{document}